\documentclass{easychair}

\usepackage{doc}

\usepackage{color,hyperref}
\usepackage{fancyvrb}
\usepackage{eurosym}
\usepackage{titlesec}
\usepackage{enumitem}
\usepackage{tabularx}
\usepackage{lastpage}
\usepackage{parskip}
\usepackage{graphicx}
\usepackage{subfigure}
\usepackage{psfrag}
\usepackage{amsmath,amssymb,amsfonts}
\usepackage{mathrsfs}
\usepackage{cite}
\usepackage{pifont} 
\usepackage{caption}
\usepackage{booktabs}
\usepackage{siunitx} 
\usepackage{algorithm}  
\usepackage{algpseudocode}

\newtheorem{lemma}{Lemma}
\newtheorem{definition}{Definition}[]
\newtheorem{proposition}{Proposition}[]
\newtheorem{theorem}{Theorem}[]

\renewcommand{\^}[1]{^{(#1)}}
\renewcommand{\th}[1]{#1^\text{th}}
\DeclareMathOperator{\sign}{\mathtt{sgn}}
\def\O{\mathcal{O}} 

\title{On Computing the Minkowski Difference of Zonotopes}

\author{
Matthias Althoff
}

\institute{
  Technical University of Munich, \\
  Department of Informatics, \\
  Munich, Germany\\
  \email{althoff@in.tum.de}
 }

\authorrunning{M.~Althoff}

\titlerunning{On Computing the Minkowski Difference of Zonotopes}

\begin{document}

\maketitle

\begin{abstract}
Zonotopes are becoming an increasingly popular set representation for formal verification techniques. This is mainly due to their efficient representation and their favorable computational complexity of important operations in high-dimensional spaces. In particular, zonotopes are closed under Minkowski addition and linear maps, which can be very efficiently implemented. Unfortunately, zonotopes are not closed under Minkowski difference for dimensions greater than two. However, we present algorithms that efficiently compute under-approximations and over-approximations of the Minkowski difference in generator representation. The efficiency of the proposed solutions is demonstrated by numerical experiments. These experiments show a reduced computation time in comparison to first computing the halfspace representation of zonotopes followed by computing their Minkowski difference. 
\end{abstract}

\section{Introduction} \label{sec:introduction}

Zonotopes have recently enjoyed a lot of popularity as a set representation for formal methods in engineering and computer science. One of the main reasons is that zonotopes can be efficiently stored in computer systems, especially when dealing with high-dimensional problems. Another important reason is that zonotopes are closed under linear maps and Minkowski addition as shown in Sec. \ref{sec:preliminariesAndProblemStatement}. 
We first review existing literature on zonotopes for formal verification and state estimation of continuous dynamic systems, formal verification of computer programs, and problems in computational geometry.

Today, zonotopes are widely used to compute the reachable set of continuous dynamic systems, i.e., the set of states that are reachable by the solution of a differential equation when the initial state, inputs, and parameters are uncertain within bounded sets \cite{Althoff2021b}. Early works on this problem used a variety of set representations, such as polytopes \cite{Chutinan2003}, ellipsoids \cite{Kurzhanski2000}, and level sets \cite{Mitchell2005}. All the mentioned set representations are either not closed under important operations (e.g. ellipsoids and oriented hyper-rectangles are not closed under Minkowski addition) or are computationally inefficient in high-dimensional spaces (e.g. polytopes and level sets). For linear continuous systems in particular, zonotopes provide an excellent compromise between accuracy and efficiency as first demonstrated by K{\"u}hn \cite{Kuehn1998}. Later, Girard \cite{Girard2005} extended the approach to uncertain inputs and continuous time. It has led to a wrapping-free algorithm in \cite{Girard2006}, i.e., an algorithm for which over-approximations are not accumulating. Further extensions of the previously mentioned work are the use for systems with uncertain parameters \cite{Althoff2011a}, nonlinear ordinary differential equations \cite{Althoff2008c, Althoff2013a}, and nonlinear differential-algebraic systems \cite{Althoff2014a}. Reachable sets have been applied to many technical realizations as summarized in \cite{Althoff2021b}.

Zonotopes are also used to rigorously estimate the states of dynamical systems as an alternative to observers that optimize with respect to the best estimate, such as Kalman filters. One of the first works that uses zonotopes for state-bounding observers is \cite{Combastel2003}. As with reachability analysis, this work has been extended to nonlinear systems in \cite{Alamo2003,Combastel2005} and systems with uncertain parameters \cite{Le2013}. More recent work is compared and surveyed in \cite{Althoff2021c}. A further application of zonotopes for continuous dynamic systems is model predictive control with guaranteed stability \cite{Bravo2006}. Advances in using zonotopes for control and guaranteed state estimation are summarized in \cite{Le2013b}.

In computer science, zonotopes are used as abstract domains in abstract interpretation for static program analysis \cite{Gaubolt2006}, whose implementation details can be found in \cite{Ghorbal2009}. An extension of zonotopes with (sub-)polyhedric domains can be found in \cite{Ghorbal2010}. Further extensions of zonotopes exist, but for the brevity of the literature review, they are not presented. Zonotopes are also used in automated theorem provers as a set representation \cite{Immler2015a}. Finally, zonotopes are used as bounding volumes to facilitate fast collision detection algorithms \cite{Guibas2005}.

Zonotopes are also an active research area in computational geometry. However, this paper mainly focuses on computational aspects on Minkowski difference targeted for applications in engineering and computer science. Thus, recent developments regarding combinatorics and relations to other mathematical problems are only briefly reviewed. The association of zonotopes with higher Bruhat orders is described in \cite{Felsner2001}. Coherence and enumeration of tilings of 3-zonotopes are addressed in \cite{Bailey1999}. Properties of zonotopes with large 2D-cuts are derived in \cite{Roerig2009} with examples provided by the \textit{Ukrainian Easter eggs}. In \cite{Campi1994}, an enclosure of zonoids by zonotopes is derived, which has the same support values for fixed directions. A bound for the number of generators with equal length of a zonotope required to enclose a ball up to a certain Hausdorff distance is obtained in \cite{Bourgain1993}. In \cite{Ferrez2005} it is shown that the problem of maximizing a quadratic form in $n$ binary variables when the underlying (symmetric) matrix is positive semidefinite, can be reduced to searching the extreme points of a zonotope. The problem of listing all extreme points of a zonotope is addressed in \cite{Fukuda2004}.

As shown above, the applications of zonotopes are manifold. Most works use zonotopes to enclose other sets, resulting in over-approximations. However, for many applications the ability to compute the Minkowski difference\footnote{Note that the term \textit{Pontryagin difference} is often used as a synonym for Minkowski difference. However, we use the term \textit{Minkowski difference} throughout this paper.} is essential. The use of the Minkowski difference can be exemplified through the computation of invariant sets of dynamic systems \cite{Lombardi2011}, reachability analysis \cite{Rakovic2006}, robust model predictive control \cite{Richards2007}, optimal control \cite{Kerrigan2002}, robotic path planning \cite{Bernabeu2001}, robust interval regression analysis \cite{Inuiguchi2002}, and cooperative games \cite{Danilov2000}. Providing an algorithm for computing the Minkowski difference of zonotopes would open up many new possibilities. Minkowski difference is well-known and rather straightforward to implement for polytopes for which open source implementations exist, see e.g. \cite{Herceg2013}. Since a zonotope is a special case of a polytope, one could use the algorithms for polytopes. However, this is less efficient compared to the novel computation for zonotopes as presented later.

A preliminary version of this article is available on arXiv \cite{Althoff2016e}. The previous version, however, could only compute approximations of the Minkowski sum. In this article, we will provide under-approximations and over-approximations, which are crucial for formal analyses. An under-approximation based on constrained zonotopes has recently been presented in \cite{Raghuraman2022}. However, the approaches presented in this work have computational advantages in our numerical experiments. To the best knowledge of the author, there exists no published algorithm for computing an over-approximative Minkowski difference of zonotopes. Further extensions to \cite{Althoff2016e} are that we provide a) order reduction techniques for computing the Minkowski difference, b) methods for degenerate sets, and c) show that the Minkowski difference is exact for aligned zonotopes.

The paper is organized as follows: We recall in Sec.~\ref{sec:preliminariesAndProblemStatement} some preliminaries on zonotopes and provide the definition of the Minkowski difference. In Sec.~\ref{sec:halfspaceConversion}, algorithms are presented for computing the halfspace representation of the Minkowski difference of zonotopes. The resulting halfspace representation is used to obtain under-approximations and over-approximations of the Minkowski difference in generator representation in Sec.~\ref{sec:generatorAdjustment}. Special cases for which the Minkowski difference can be computed exactly are presented in Sec.~\ref{sec:exactSolution}. The performance of tour approach is demonstrated by numerical experiments in Sec.~\ref{sec:numericalExperiments}.

\section{Preliminaries and Problem Statement} \label{sec:preliminariesAndProblemStatement}

Throughout this paper, we index elements of vectors and matrices by subscripts and enumerate vectors or matrices by superscripts in parentheses to avoid confusion with the exponentiation of a variable. For instance, $A\^k_{ij}$ is the element of the $\th{i}$ row and $\th{j}$ column of the $\th{k}$ matrix $A\^k$. We first recall the representation of a zonotope. 

\begin{definition}[Zonotope (G-Representation)]
Zonotopes are sets parameterized by a center $c \in \mathbb{R}^n$ and generators $g\^i \in \mathbb{R}^n$:
\begin{equation} \label{eq:zonotope}
 \mathcal{Z} = \Big\{ c + \sum_{i=1}^{p} \beta_i \, g\^{i} \Big| \beta_i \in [-1,1] \Big\}.
\end{equation}
 The order of a zonotope is defined as $\varrho = \frac{p}{n}$.
\end{definition}
We write in short $\mathcal{Z}=(c,g\^{1},\ldots,g\^{p})$. Zonotopes are a compact way of representing sets in high-dimensional spaces. More importantly, linear maps $M\otimes \mathcal{Z}:=\{M z |z \in \mathcal{Z}\}$ ($M\in\mathbb{R}^{q\times n}$) and Minkowski addition $\mathcal{Z}_1\oplus \mathcal{Z}_2 := \{z_1 + z_2 | z_1 \in \mathcal{Z}_1 , z_2 \in \mathcal{Z}_2\}$, as required in many of the applications mentioned in Sec. \ref{sec:introduction}, can be computed efficiently and exactly \cite{Kuehn1998b}. Given $\mathcal{Z}_1=(c, g\^{1}, \ldots, g\^{p_1})$ and $\mathcal{Z}_2=(d, h\^{1}, \ldots, h\^{p_2})$ one can efficiently compute
\begin{equation} \label{eq:additionAndMultiplication}
\begin{split}
 \mathcal{Z}_1 \oplus \mathcal{Z}_2 & = (c + d, g\^{1}, \ldots, g\^{p_1}, h\^{1}, \ldots, h\^{p_2}), \\
 M\otimes \mathcal{Z}_1 & = (M \, c, M\, g\^{1}, \ldots, M \, g\^{p_1}).
\end{split}
\end{equation}
Note that in the remainder of this paper, the symbol for set-based multiplication is often omitted for simplicity of notation. Further, one or both operands of set-based operations can be singletons and set-based multiplication has precedence over Minkowski addition. 
A zonotope can be interpreted in three ways, see e.g. \cite[p. 364]{Henk2004} and \cite[Sec.~7.3]{Ziegler1995}. All interpretations are now introduced, as we use each one in this paper to show certain properties concisely.

\paragraph{\textbf{Minkowski addition of line segments (first interpretation)}} A zonotope can be interpreted as the Minkowski addition of line segments $l\^i = [-1,1]\, g\^{i}$, which is visualized step-by-step for $\mathbb{R}^2$ in Fig. \ref{fig:zonotope}. 

\paragraph{\textbf{Projection of a hypercube (second interpretation)}} A zonotope can be interpreted as the projection of a $p$-dimensional unit hypercube $\mathcal{C} = [-1,1]^p$ onto the $n$-dimensional space by the matrix of generators $G = \begin{bmatrix} g\^{1}, & \ldots, & g\^{p} \end{bmatrix}$, which is then translated to the center $c$: $\mathcal{Z} = c \oplus G \otimes \mathcal{C}$. We write in short $\mathcal{Z} = (c, G)$.

\paragraph{\textbf{Polytopes whose faces are centrally symmetric (third interpretation)}} A zonotope can be interpreted as a polytope whose $j$-faces are centrally symmetric. As later shown in Sec.~\ref{sec:halfspaceConversion}, the facets ($(n-1)$-faces) are obtained by choosing particular halfspaces $\{x \in \mathbb{R}^n \big| {c\^{i}}^T \, x \leq d_i\}$, $c\^{i}\in\mathbb{R}^{n}$, $d_i\in\mathbb{R}$ whose intersection forms the halfspace representation (H-representation) of a zonotope:
\begin{equation} \label{eq:zonotopeHalfspace}
 \mathcal{Z}_H = \Big\{x \in \mathbb{R}^n \big| C \, x \leq d \Big\},
\end{equation}
where $C = \begin{bmatrix} c\^1, & \ldots, & c\^q \end{bmatrix}^T$ and $d = \begin{bmatrix} d_1, & \ldots, & d_q \end{bmatrix}^T$.

\begin{figure}[htb]	

  \footnotesize
  \psfrag{#c}[][]{$c$}						
  \psfrag{#l1}[][]{$l\^1$}	
  \psfrag{#l2}[][]{$l\^2$}	
  \psfrag{#l3}[][]{$l\^3$}					
    \centering		
      \subfigure[$c\oplus l\^1$]
		{\includegraphics[width=0.2\columnwidth]{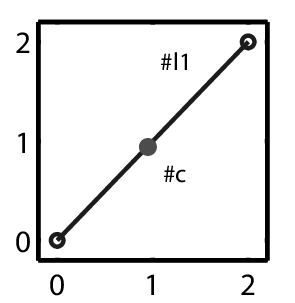}}
      \subfigure[$c\oplus l\^1 \oplus l\^2$]
		{\includegraphics[width=0.2\columnwidth]{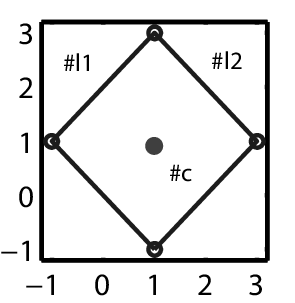}}			
      \subfigure[$c\oplus \ldots \oplus l\^3$]
		{\includegraphics[width=0.2\columnwidth]{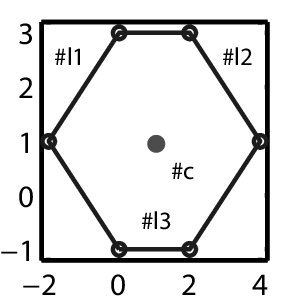}}					
      \caption{Step-by-step construction of a zonotope.}
      \label{fig:zonotope}
\end{figure}

Related to the Minkowski addition is the Minkowski difference. Given the minuend $\mathcal{Z}_m$ and the subtrahend $\mathcal{Z}_s$ of equal dimension, the Minkowski difference is defined as (see \cite{Montejano1996})
\begin{equation} \label{eq:MinkowskiDifference}
 \mathcal{Z}_m\ominus \mathcal{Z}_s = \{x \in \mathbb{R}^n | x \oplus \mathcal{Z}_s \subseteq \mathcal{Z}_m \},
\end{equation}
such that $(\mathcal{Z}_m\ominus \mathcal{Z}_s) \oplus \mathcal{Z}_s \subseteq \mathcal{Z}_m$. We refer to $\mathcal{Z}_m\ominus \mathcal{Z}_s$ as the \textit{difference}. When $\tilde{\mathcal{Z}}_m = \mathcal{A} \oplus \mathcal{Z}_s$, where $\mathcal{A}$ is an arbitrary set, we have $(\tilde{\mathcal{Z}}_m \ominus \mathcal{Z}_s) \oplus \mathcal{Z}_s = \tilde{\mathcal{Z}}_m$. An alternative definition (see \cite{Ghosh1991,Montejano1996}) is 
\begin{equation} \label{eq:infiniteIntersectionDefinition}
 \mathcal{Z}_m\ominus \mathcal{Z}_s = \bigcap_{z_s \in \mathcal{Z}_s} (\mathcal{Z}_m - z_s).
\end{equation}
The Minkowski difference can be obtained by translations along generators in Theorem~\ref{thm:minkowskiDifferenceFromGenerators}, which is based on Lemma~\ref{thm:finitelyManyIntersections}.

\begin{lemma}[Minkowski Difference from Finitely Many Intersections \mbox{\cite[Theorem 2.1]{Kolmanovsky1998a}}]
\label{thm:finitelyManyIntersections}
 When the subtrahend $\mathcal{Z}_s$ is convex, it suffices to compute the Minkowski difference from intersections of translations by the vertices $v\^i \in \mathcal{V}$ of $\mathcal{Z}_s$:
 \begin{equation*} 
 \mathcal{Z}_m\ominus \mathcal{Z}_s = \bigcap_{v\^i \in \mathcal{V}} (\mathcal{Z}_m - v\^i)
 \end{equation*}
\end{lemma}
\begin{proof}
 See \cite[Theorem 2.1]{Kolmanovsky1998a}.
\end{proof}

\begin{theorem}[Minkowski Difference from Generators] \label{thm:minkowskiDifferenceFromGenerators}
 When the sets $\mathcal{Z}_m$ and $\mathcal{Z}_s$ are zonotopes, it suffices to apply the following recursive procedure using only the generators $g\^{s,i}$ of $\mathcal{Z}_s = (c\^s, g\^{s,1}, \ldots, g\^{s,p_s})$ to obtain $\mathcal{Z}_m \ominus \mathcal{Z}_s$:
 \begin{equation*} 
 \begin{split}
 \mathcal{Z}_{int}\^1 =& \mathcal{Z}_m - c\^s \\
 \forall i=1\ldots p_s: \quad \mathcal{Z}_{int}\^{i+1} =& (\mathcal{Z}_{int}\^{i} + g\^{s,i}) \cap (\mathcal{Z}_{int}\^{i} - g\^{s,i}) \\
 \mathcal{Z}_m\ominus \mathcal{Z}_s =& \mathcal{Z}_{int}\^{p_s+1}
 \end{split}
 \end{equation*}
\end{theorem}
\begin{proof}
 As shown in \cite[eq.~2]{Danilov2000}, it generally holds for sets $\mathcal{A}$, $\mathcal{B}$, and $\mathcal{C}$ that
 \begin{equation} \label{eq:MinkowskiAdditionToDifference}
  \mathcal{A} \ominus (\mathcal{B} \oplus \mathcal{C}) = (\mathcal{A} \ominus \mathcal{B}) \ominus \mathcal{C}.
 \end{equation}
 Let us rewrite $\mathcal{Z}_m \ominus \mathcal{Z}_s = \mathcal{Z}_m \ominus (c\^s \bigoplus_{i=1}^{p_s} [-1,1] \otimes g\^{s,i})$. By recursively applying \eqref{eq:MinkowskiAdditionToDifference} we obtain
 \begin{equation} \label{eq:recursiveMinkoskiDifference}
  \mathcal{Z}_m \ominus \mathcal{Z}_s = \bigg(\Big((\mathcal{Z}_m - c\^s) \ominus [-1,1] \otimes g\^1 \Big) \ominus \ldots \bigg) \ominus [-1,1] \otimes g\^{p_s}.
 \end{equation}
 The Minkowski difference with $[-1,1] \otimes g\^{s,i}$ can be further simplified according to Lemma~\ref{thm:finitelyManyIntersections} by only considering the extreme cases, such that for a set $\mathcal{A}$ we have
 \begin{equation} \label{eq:generatorDifferenceToIntersection}
  \mathcal{A} \ominus [-1,1] \otimes g\^{s,i} = (\mathcal{A} + g\^{s,i}) \cap (\mathcal{A} - g\^{s,i}).
 \end{equation}
 Inserting \eqref{eq:generatorDifferenceToIntersection} into \eqref{eq:recursiveMinkoskiDifference} results in the theorem to be proven.
\end{proof}

To provide the reader with a better understanding of the Minkowski difference of zonotopes, we show three distinctive examples in Fig.~\ref{fig:minkowskiDifferenceExamples}: (a) the zonotope order of $\mathcal{Z}_d = \mathcal{Z}_m \ominus \mathcal{Z}_s$ equals the one of the minuend, (b) the order is reduced, and (c) the result is the empty set. We choose
\begin{equation*}
\begin{split}
 & \mathcal{Z}_m = \left( \begin{bmatrix} 1 \\ 1 \end{bmatrix}, \begin{bmatrix} 1 \\ 0 \end{bmatrix}, \begin{bmatrix} 0 \\ 1 \end{bmatrix}, \begin{bmatrix} 1 \\ 1 \end{bmatrix} \right), \\
 & \mathcal{Z}_{s,1} = \left( \begin{bmatrix} 0 \\ 0 \end{bmatrix}, \begin{bmatrix} 0.5 \\ -0.2 \end{bmatrix}, \begin{bmatrix} 0 \\ 0.2 \end{bmatrix} \right), 
 \mathcal{Z}_{s,2} = \left( \begin{bmatrix} 0 \\ 0 \end{bmatrix}, \begin{bmatrix} 0.5 \\ -0.5 \end{bmatrix}, \begin{bmatrix} 0 \\ 0.5 \end{bmatrix} \right),
 \mathcal{Z}_{s,3} = \left( \begin{bmatrix} 0 \\ 0 \end{bmatrix}, \begin{bmatrix} 2 \\ -0.5 \end{bmatrix}, \begin{bmatrix} 0 \\ 0.5 \end{bmatrix} \right).
\end{split}
\end{equation*}

\begin{figure}[htb]	

  \footnotesize
  \psfrag{a}[c][c]{$x_1$}						
  \psfrag{b}[c][c]{$x_2$}	
  \psfrag{c}[r][c]{$\mathcal{Z}_d \oplus \mathcal{Z}_s$}	
  \psfrag{d}[l][c]{$\mathcal{Z}_m$}			
  \psfrag{e}[l][c]{$\mathcal{Z}_s$}	
  \psfrag{f}[r][c]{$\mathcal{Z}_d$}	
    \centering		
      \subfigure[$\mathcal{Z}_d = \mathcal{Z}_m \ominus \mathcal{Z}_{s,1}$; result has the same order as $\mathcal{Z}_m$.]
		{\includegraphics[width=0.32\columnwidth]{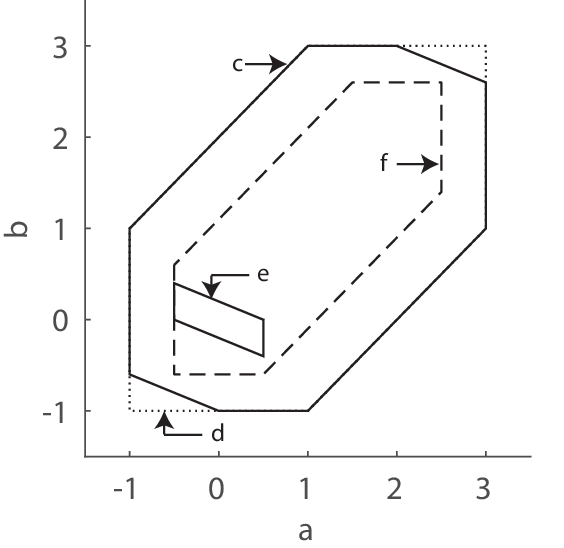} \label{fig:minkowskiDiffFull}}
	\psfrag{e}[c][c]{$\mathcal{Z}_s$}	
  \psfrag{f}[c][c]{$\mathcal{Z}_d$}	
      \subfigure[$\mathcal{Z}_d = \mathcal{Z}_m \ominus \mathcal{Z}_{s,2}$; result has a smaller order than $\mathcal{Z}_m$.]
		{\includegraphics[width=0.32\columnwidth]{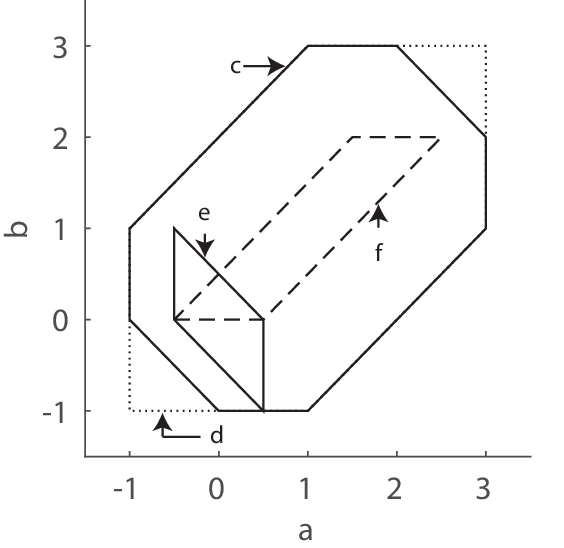} \label{fig:minkowskiDiffRed}}	
	\psfrag{d}[r][c]{$\mathcal{Z}_m$}
      \subfigure[$\mathcal{Z}_d = \mathcal{Z}_m \ominus \mathcal{Z}_{s,3}$; result is empty.]
		{\includegraphics[width=0.32\columnwidth]{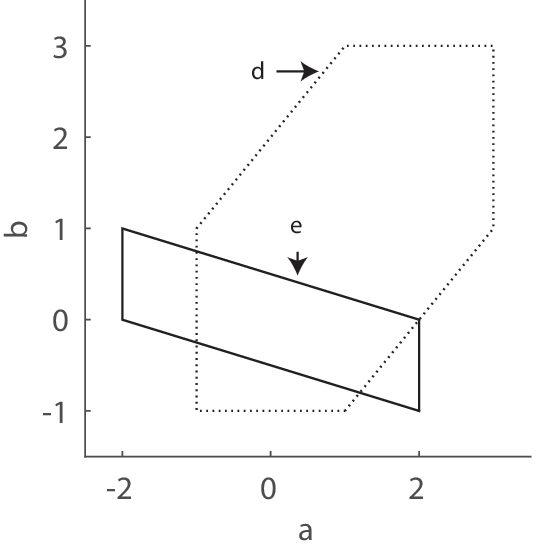} \label{fig:minkowskiDiffEmpty}}					
      \caption{Results of different Minkowski differences.}
      \label{fig:minkowskiDifferenceExamples}
\end{figure}
One can observe in Fig.~\ref{fig:minkowskiDiffFull} that all halfspaces of $\mathcal{Z}_m$ remain for $\mathcal{Z}_m \ominus \mathcal{Z}_{s,1}$. The result of $\mathcal{Z}_m \ominus \mathcal{Z}_{s,2}$ in Fig.~\ref{fig:minkowskiDiffRed} does not require all halfspaces of $\mathcal{Z}_m$. Finally, Fig.~\ref{fig:minkowskiDiffEmpty} shows that $\mathcal{Z}_m \ominus \mathcal{Z}_{s,3} = \emptyset$. The halfspace representation of the Minkowski difference of zonotopes is addressed in the next section.

\section{Halfspace Representation of the Minkowski Difference} \label{sec:halfspaceConversion}

As it is later shown in Sec.~\ref{sec:generatorAdjustment}, zonotopes are not closed under Minkowski difference (unless in the two-dimensional case and when generators are aligned). One possibility to obtain the Minkowski difference is to first convert both zonotopes into halfspace representation and then use standard algorithms for obtaining the Minkowski difference. In this section, we propose a more efficient algorithm for obtaining the halfspace representation of the Minkowski difference. For that purpose, we require the $n$-dimensional cross product, which is an extension of the well-known cross product of two three-dimensional vectors. The resulting vector is orthogonal to all $n-1$ $n$-dimensional vectors.

\begin{definition}[$n$-dimensional cross product (see \cite{Mortari1996})] \label{def:nDimCrossProduct}
Given are $n-1$ vectors $h\^i\in\mathbb{R}^n$ which are stored in a matrix $H=[h\^1,\ldots,h\^{n-1} ]\in\mathbb{R}^{n\times n-1}$. We further denote by $H^{[i]}\in\mathbb{R}^{n-1\times n-1}$ the matrix $H$, where the $\th{i}$ row is removed. The cross product $nX(H)$ of the vectors stored in $H$ is defined as
\begin{equation*}
	y=nX(H)=\begin{bmatrix} \det(H^{[1]}), & \ldots, & (-1)^{i+1}\det(H^{[i]}), & \ldots, & (-1)^{n+1}\det(H^{[n]}) \end{bmatrix}^T.
\end{equation*}
\end{definition}

From now on, we assume that all generators of each zonotope are not aligned. If two aligned generators would exist, e.g. $\gamma \, g\^i = g\^j$, $\gamma \in \mathbb{R}$, one could adjust $g\^i := (\gamma + 1) g\^i$ and remove $g\^j$. For a general zonotope, the generator matrix $G$ is of dimension $n \times p$. The normal vector of each facet is obtained from the $n$-dimensional cross product of $n-1$ generators, which have to be selected from $p$ generators for each non-parallel facet, so that one obtains $2\binom{p}{n-1}$ facets \cite[Lemma 3.1]{Gover2010}. This is always possible since we assume without loss of generality that all generators are not aligned. The generators that span a facet are obtained by canceling $p-n+1$ generators from the generator matrix $G$, which is denoted by $G^{\langle \gamma,\ldots,\eta \rangle}$, where $\gamma,\ldots,\eta$ are the $p-n+1$ indices of the generators that are taken out of $G$. 
\begin{theorem}[H-Representation of Zonotopes \mbox{\cite[Theorem.~7]{Althoff2010d}}] \label{thm:hZonotope}
The halfspace representation $C\, x\leq d$ of a zonotope $(c,G)$ with $p$ independent generators is
\begin{gather*}
	C=\begin{bmatrix} C^+ \\ -C^+\end{bmatrix}, \quad C^+= \begin{bmatrix} C_1^+ \\ \vdots \\ C_\nu^+ \end{bmatrix}, \quad C_i^+=\frac{nX(G^{\langle \gamma,\ldots,\eta \rangle})^T}{\|nX(G^{\langle \gamma,\ldots,\eta \rangle})\|_2}, \\
	d=\begin{bmatrix} d^+ \\ d^-\end{bmatrix} = \begin{bmatrix} {C^+}\, c + \Delta d \\ -{C^+}\, c + \Delta d \end{bmatrix}, \quad
	\Delta d = \sum_{\upsilon=1}^{p} |{C^+}\, g\^{\upsilon}|. 
\end{gather*}
The index $i$ varies from $1$ to $\nu = \binom{p}{n-1}$. 
\end{theorem}
\begin{proof}
The following proof is more intuitive than that of \cite[Theorem.~7]{Althoff2010d} and provides valuable insights for the remainder of the paper.
The $\th{i}$ facet is spanned by $n-1$ generators, which are obtained by canceling $p-n+1$ generators with indices $\gamma,\ldots,\eta$ from $G$, which is denoted by $G^{\langle \gamma,\ldots,\eta \rangle}$ \cite[Sec.~3]{Gover2010}. This is illustrated for a two-dimensional example in Fig.~\ref{fig:generatorsReachingFacets}. As illustrated in Fig.~\ref{fig:normalVector}, the normal vector of this facet is obtained by the normalized $n$-dimensional cross product (see Def. \ref{def:nDimCrossProduct}):
\begin{equation*}
 C_i^+=nX(G^{\langle \gamma,\ldots,\eta \rangle})^T/\|nX(G^{\langle \gamma,\ldots,\eta \rangle})\|_2. 
\end{equation*}
It is sufficient to compute $\nu$ normal vectors denoted by a superscript '$+$', as the remaining $\nu$ normal vectors denoted by a superscript '$-$' differ only in sign due to the central symmetry of zonotopes. A possible point $x\^i$ on the $\th{i}$ halfspace is obtained by adding generators in the direction of the normal vector to the center (see Fig.~\ref{fig:generatorsReachingFacets}): 
\begin{equation*}
 x\^i = c + \sum_{\upsilon=1}^{p} \sign(C_i^+ g\^{\upsilon}) g\^{\upsilon},
\end{equation*}
where $\sign()$ is the sign function returning the sign of a value. Note that the generators spanning the $\th{i}$ facet are not required to reach the halfspace. To keep the result simple, however, we add them in the above formula, since they only translate $x\^i$ on the facet. Since the elements $d_i^+$ are the scalar products of any point on the $\th{i}$ halfspace with its normal vector, we obtain 
\begin{equation*}
 \begin{split}
  d_i^+ = {C_i^+} x\^i 
        = {C_i^+} \Big(c + \sum_{\upsilon=1}^{p} \sign({C_i^+}\, g\^{\upsilon}) g\^{\upsilon}\Big) 
        = {C_i^+} c + \sum_{\upsilon=1}^{p} |{C_i^+}\, g\^{\upsilon}| = {C_i^+} c + \Delta d_i.
 \end{split}
\end{equation*}
Analogously, the values for $d_i^-$ are obtained. 
\end{proof}

\begin{figure}[h!tb]
\begin{minipage}{0.4\columnwidth}
  \centering
  \psfrag{a}[c][c]{$C_1^+$}
  \psfrag{b}[c][c]{$C_2^+$}
	\psfrag{c}[c][c]{$C_3^+$}
	\psfrag{d}[c][c]{$C_4^+$}
	\psfrag{e}[c][c]{$c$}
  \psfrag{f}[c][c]{$g\^1$}
	\psfrag{g}[c][c]{$g\^2$}
	\psfrag{h}[c][c]{$g\^3$}
  \psfrag{i}[c][c]{$g\^4$}
	\psfrag{j}[c][c]{$-g\^1$}
	\psfrag{k}[c][c]{$-g\^2$}
  \psfrag{l}[c][c]{$g\^3$}
	\psfrag{m}[c][c]{$g\^2$}
	\psfrag{n}[c][c]{$g\^3$}
	\psfrag{o}[l][c]{$x\^1=x\^2$}
	\psfrag{q}[l][c]{$x\^3$}
	\psfrag{r}[l][c]{$x\^4$}
    \includegraphics[width=\columnwidth]{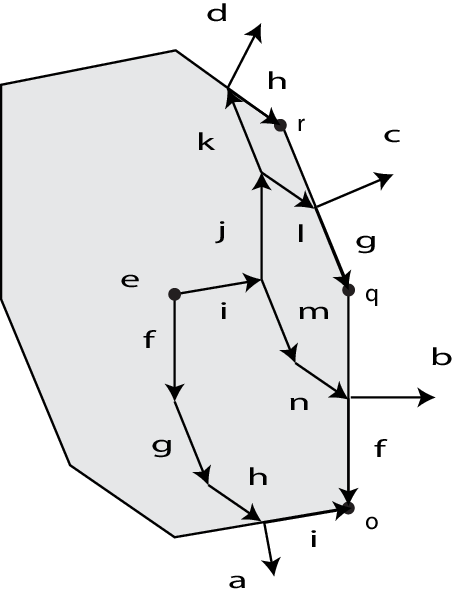}
  \caption{Various generator additions to reach corresponding facets in two dimensions. Only the generators not spanning the facet are required to reach the facet.}
  \label{fig:generatorsReachingFacets}
\end{minipage}
\hspace{0.08\columnwidth}
\begin{minipage}{0.5\columnwidth}
  \centering
  \psfrag{a}[c][c]{$0$}
  \psfrag{b}[c][c]{$-2$}
	\psfrag{c}[c][c]{$2$}
	\psfrag{x}[c][c]{$x_1$}
	\psfrag{y}[c][c]{$x_2$}
  \psfrag{z}[c][c]{$x_3$}
	\psfrag{d}[c][c]{$g\^1$}
	\psfrag{e}[c][c]{$g\^2$}
  \psfrag{f}[c][c]{$C_1^+$}
    \includegraphics[width=\columnwidth]{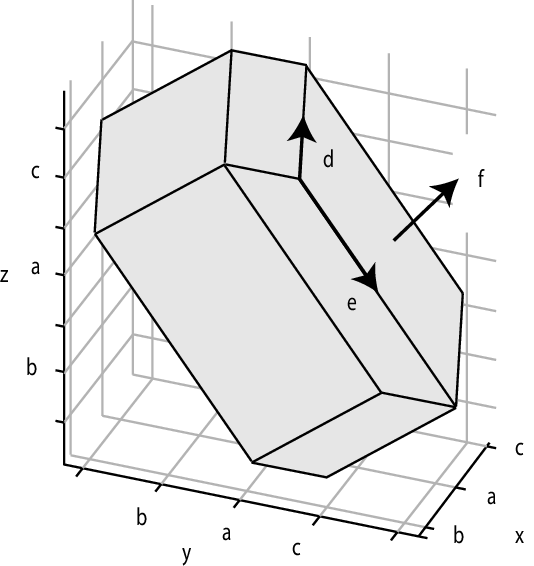}
  \caption{Normal vector of a facet of a three-dimensional zonotope spanned by two generators.}
  \label{fig:normalVector}
\end{minipage}
\end{figure}

%
Next, we directly obtain the halfspace representation of the intersection of two zonotopes, which are identical, except that one of them is translated by a vector $2 \, h$.

\begin{lemma}[H-Representation of Intersection of Translated Zonotopes] \label{thm:hZonotopeIntersection}
Given are two zonotopes $\mathcal{Z}_o = (c-h,G)$ and $\mathcal{Z}_t = (c+h,G)$. The halfspace representation $C\, x\leq d$ of the intersection $\mathcal{Z}_o \cap \mathcal{Z}_t$ has an identical $C$ matrix as presented in Theorem~\ref{thm:hZonotope}, but a changed $d$ vector
\begin{gather}\label{eq:translatedZonotope}
	d=\begin{bmatrix} d^+ \\ d^-\end{bmatrix} = \begin{bmatrix} {C^+}\, c + \Delta d - \Delta d_{trans} \\ -{C^+}\, c + \Delta d - \Delta d_{trans} \end{bmatrix}, \quad
	\Delta d = \sum_{\upsilon=1}^{p} |{C^+}\, g\^{\upsilon}|, \quad \Delta d_{trans} = |{C^+} h|.
\end{gather}
\end{lemma}
\begin{proof}
 Since the translated zonotope $\mathcal{Z}_t$ has the same generators as the original zonotope $\mathcal{Z}_o$, both halfspace representations of $\mathcal{Z}_t$ and $\mathcal{Z}_o$ will have the same normal vectors, see Theorem~\ref{thm:hZonotope}. For each normal vector $C_i^+$, one of the corresponding halfspaces from $\mathcal{Z}_o$ or $\mathcal{Z}_t$ is a subset of the other one as shown in Fig.~\ref{fig:redundantGenerators}. This, of course, is analogous for $C_i^-$. Consequently, the number of required normal vectors for the intersection is unchanged compared to Theorem~\ref{thm:hZonotope} as also illustrated in Fig.~\ref{fig:redundantGenerators}. 
 
However, the intersection results in a subset of $\mathcal{Z}_o$, which is equivalent in reducing the values of $d_i^+$ and $d_i^-$ by $\Delta d_{trans,i}$. From the duality of intersection of translated sets and Minkowski difference (see \eqref{eq:infiniteIntersectionDefinition}) and after introducing $0_n$ as the $n$-dimensional vector of zeros, we have that
 \begin{equation}\label{eq:auxiliaryEquationTranslatedZonotope}
  (0_n,h) \oplus (\mathcal{Z}_o \cap \mathcal{Z}_t) = (0_n,h) \oplus \Big(\underbrace{(c-h,G) \cap (c+h,G)}_{(c,G) \ominus (0_n,h)}\Big) = (c,G).
 \end{equation}
 The Minkowski addition of $(0_n,h)$ in \eqref{eq:auxiliaryEquationTranslatedZonotope} can be seen as an additional generator so that $\Delta d = \sum_{\upsilon=1}^{p} |{C^+}\, g\^{\upsilon}| + |{C^+}\, h|$ in \eqref{eq:translatedZonotope}. To compensate for the change in $\Delta d$, we choose $\Delta d_{trans} = |{C^+} h|$ in \eqref{eq:translatedZonotope}.
\end{proof}

\begin{figure}[h!tb]
  \centering
  \footnotesize
  \psfrag{a}[c][c]{$C_1^+$}
  \psfrag{b}[c][c]{$-C_1^+$}
	\psfrag{c}[c][c]{$c$}
	\psfrag{d}[c][c]{$\mathcal{Z}_o$}
	\psfrag{e}[c][c]{$\mathcal{Z}_t$}
  \psfrag{f}[c][c]{$\mathcal{Z}_o \cap \mathcal{Z}_t$}
	\psfrag{h}[c][c]{$h$}
    \includegraphics[width=0.8\columnwidth]{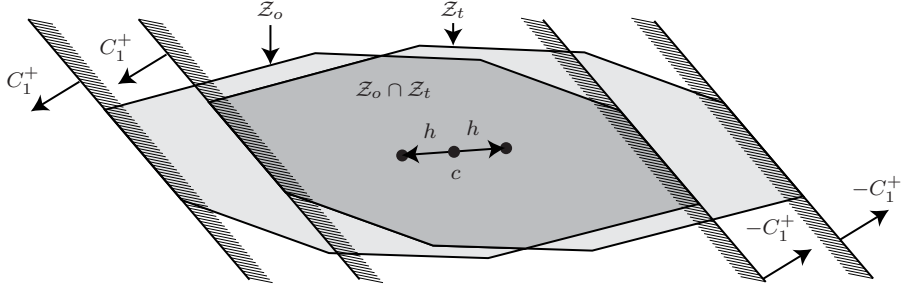}
  \caption{Either the halfspace belonging to $C_1^+$ or $C_1^-$ is redundant for $\mathcal{Z}_o$ when computing the interaction $\mathcal{Z}_o \cap \mathcal{Z}_t$. The same applies to $\mathcal{Z}_t$}
  \label{fig:redundantGenerators}
\end{figure}
%
The above lemma is used to compute the halfspace representation of $\mathcal{Z}_m\ominus \mathcal{Z}_s$.

\begin{theorem}[H-Representation of the Minkowski Difference] \label{thm:hMinkowskiDifference}
Given are the two zonotopes $\mathcal{Z}_m = (c\^m,G\^m)$ and $\mathcal{Z}_s = (c\^s,G\^s)$. A halfspace representation $C\, x\leq d$ of the Minkowski difference $\mathcal{Z}_m \ominus \mathcal{Z}_s$ has an identical $C$ matrix as Theorem~\ref{thm:hZonotope}, but a changed $d$ vector:
\begin{gather} 
 d=\begin{bmatrix} d^+ \\ d^-\end{bmatrix} = \begin{bmatrix} {C^+}\, (c\^m -c\^s) + \Delta d - \Delta d_{trans} \\ -{C^+}\, (c\^m -c\^s) + \Delta d - \Delta d_{trans} \end{bmatrix}, \notag \\
	\Delta d = \sum_{\upsilon=1}^{p_m} |{C^+}\, g\^{m,\upsilon}|, \quad \Delta d_{trans} = \sum_{\upsilon=1}^{p_s} |{C^+}\, g\^{s,\upsilon}|. \label{eq:deltaD}
\end{gather}
The computational complexity is $\O(\binom{p}{n-1}\big( n^4 + n(p_m+p_s))\big)$.
\end{theorem}
\begin{proof}
 As shown in Theorem~\ref{thm:minkowskiDifferenceFromGenerators}, the Minkowski difference can be computed from finitely many translations of $\mathcal{Z}_m$ by the generators of $\mathcal{Z}_s$. The halfspace representation of intersecting $\mathcal{Z}_m$ with one translated version of itself has been shown in Lemma \ref{thm:hZonotopeIntersection}. Repeated application of Lemma \ref{thm:hZonotopeIntersection} results in the summation of values of $\Delta d_{trans}$ to $\Delta d_{trans} = \sum_{\upsilon=1}^{p_s} |{C^+}\, g\^{s,\upsilon}|$, thus proving the theorem.
 
 The computational complexity of Theorem~\ref{thm:hMinkowskiDifference} is obtained as follows:
We only consider the number of required binary operations; the computational effort of unary operations like
concatenations of lists can be safely neglected. The computation of the determinants is $\O((n-1)^3)=\O(n^3)$ when applying LU decomposition, and thus the computation of the $n$-dimensional cross product is $n\cdot\O(n^3)=\O(n^4)$. As $\binom{p}{n-1}$ non-parallel hyperplanes have to be computed (see Theorem~\ref{thm:hZonotope}), the computational complexity for $C^+$ is $\O(\binom{p}{n-1} n^4)$. Finally, we have $1+p_m+p_s$ multiplications of $C^+$ with generators and one difference of centers. Each multiplication has a complexity of $\O(\binom{p}{n-1}n)$ so that we obtain $\O(\binom{p}{n-1}n(p_m+p_s))$ for $d$. Thus, the overall computational complexity is $\O(\binom{p}{n-1}\big( n^4 + n(p_m+p_s))\big)$.
\end{proof}

Please note that although the result of Theorem~\ref{thm:hMinkowskiDifference} is centrally symmetric, the result is not a zonotope in general. While each zonotope is centrally symmetric, the opposite does not hold. We demonstrate this by a counterexample.
\begin{proposition}[Zonotopes are not closed under Minkowski difference] \label{thm:MinkowskiDifferenceClosedness}
 Zonotopes are not closed under Minkowski difference, i.e., given the zonotopes $\mathcal{Z}_m$, $\mathcal{Z}_s$, the set $\mathcal{Z}_d = \mathcal{Z}_m\ominus \mathcal{Z}_s$ is not a zonotope. 
 
\end{proposition}
\begin{proof}
 We prove the proposition by a counterexample with $\mathcal{Z}_m = (\mathbf{0}, G_m)$ and $\mathcal{Z}_s = (\mathbf{0}, G_s)$, where $\mathbf{0}$ is a vector of zeros and 
 \begin{equation*}
   G_m = \begin{bmatrix}
          1 & 1 & 0 & 0 \\
          1 & 0 & 1 & 0 \\
          1 & 0 & 0 & 1 
         \end{bmatrix}, \quad
   G_s = \frac{1}{3}\begin{bmatrix}
          -1 & 1 & 0 & 0 \\
          1 & 0 & 1 & 0 \\
          1 & 0 & 0 & 1 
         \end{bmatrix}.
 \end{equation*}
 As one can see in the plot of $\mathcal{Z}_m\ominus \mathcal{Z}_s$ in Fig.~\ref{fig:counterExample} obtained by CORA \cite{Althoff2015a} and the MPT toolbox \cite{Herceg2013}, not all faces are centrally symmetric as it has to be for zonotopes.
\end{proof}

\begin{figure}[htb]
  \centering
  \footnotesize
  \psfrag{a}[c][c]{$x_1$}						
  \psfrag{b}[c][c]{$x_2$}	
  \psfrag{c}[c][c]{$x_3$}							
     \includegraphics[width=0.5\columnwidth]{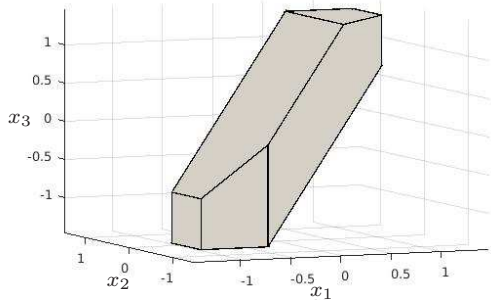}					
      \caption{The Minkowski difference of two zonotopes that is not a zonotope in general.}
      \label{fig:counterExample}
\end{figure}
In the next section, we use the halfspace representation of the Minkowski difference to obtain under-approximations and over-approximations in generator representation.

\section{Generator Removal and Contraction} \label{sec:generatorAdjustment}

As shown in Sec. \ref{sec:preliminariesAndProblemStatement}, it is possible that not all generators of the minuend $\mathcal{Z}_m$ are preserved---in the extreme case, the result is the empty set. We present how to detect the generators that can be removed, followed by approaches that contract the remaining generators so that the Minkowski difference is under-approximated or over-approximated.  

When the minuend and the subtrahend are both degenerate sets, we first compute an orthonormal basis of the column space (aka range) of the generator matrix of the minuend $(c\^m, G\^m)$. The orthonormal basis of the column space are the rows in $U$ of the singular value decomposition $G\^m = U\Sigma V^T$ whose corresponding singular values in $\Sigma$ are non-zero; these rows are stored in the matrix $P$. If $P$ is identical for the minuend and the subtrahend, the Minkowski difference is computed as $P^T\big((P\mathcal{Z}_m) \ominus (P\mathcal{Z}_s)\big)$.

\subsection{Generator Removal} \label{sec:generatorRemoval}

To remove redundant generators, we store the normal vectors of the H-representation of the minuend $\mathcal{Z}_m$ together with the indices of generators $g\^{m,i}$ ($i=1,\ldots,p_m$) spanning the corresponding facets. 
Next, the halfspace representation of the intersection is computed as presented in Theorem \ref{thm:hMinkowskiDifference} and the minimum H-representation is obtained using linear programming \cite[Sec. 2.21]{Fukuda2004b}. In general, the removal of redundant halfspaces causes the Minkowski difference to become a polytope. Exceptions exist and are discussed in Sec.~\ref{sec:exactSolution}. We denote the matrix of normal vectors of redundant halfspaces by $C^{red}$ and the one of non-redundant halfspaces by $\hat{C}$. The corresponding generator indices of the $\th{j}$ redundant halfspace $C^{red}_j$ are $[\alpha_j,\ldots, \kappa_j]$. We also introduce a vector $\mathtt{ind}\^j$ indicating whether the generator contributes to the halfspace:
\begin{equation*}
 \forall i=1,\ldots,p_m: \quad \mathtt{ind}_i\^j = \begin{cases}
                    1, \text{ if } i \in \{\alpha_j,\ldots, \kappa_j\}, \\
                    0, \text{ otherwise}. 
                   \end{cases}
\end{equation*}
The generators, which are no longer required, are selected based on $\mathtt{ind}\^j$ by the following proposition. The matrix of remaining generators is denoted by $\hat{G}\^m$.

\begin{proposition}[Generator removal]
 The $\th{i}$ generator of the minuend $\mathcal{Z}_m$ with $p_m$ generators in $n$-dimensional space is no longer required for the difference $\mathcal{Z}_m \ominus \mathcal{Z}_s$ if and only if
 \begin{equation*}
  \sum_{j=1}^{\varsigma} \mathtt{ind}_i\^j = 2 \binom{p_m-1}{n-2},
 \end{equation*}
 where $\varsigma$ is the number of redundant halfspaces. The computational complexity with respect to the dimension $n$ is $\O(\binom{p_m}{n-1} p_m^2)$.
\end{proposition}
\begin{proof}
 Each facet is spanned by $n-1$ generators. By fixing one generator for a facet, one can still select $n-2$ generators from $p_m-1$ generators, such that each generator spans $2 \binom{p_m-1}{n-2}$ facets. If a generator is removed, those $2 \binom{p_m-1}{n-2}$ facets no longer exist, such that a generator can only be removed if and only if all $2 \binom{p_m-1}{n-2}$ facets to which it contributes are redundant.
 
 The removal of a single generator only requires the summation of $\varsigma$ $p_m$-dimensional vectors and one equality check so that we have $\O(\varsigma \, p_m)$. Since the highest possible value of $\varsigma$ is $\varsigma=2\binom{p_m}{n-1}$ (see Theorem~\ref{thm:hZonotope}), the complexity becomes $\O(\binom{p_m}{n-1} p_m)$ and $\O(\binom{p_m}{n-1} p_m^2)$ for all generators.
\end{proof}
Depending on the application, one can also skip the generator removal to save computation time. However, without generator removal, the order of the resulting zonotopes is unnecessarily large.

\subsection{Generator Contraction Using Linear Programming} \label{sec:generatorContraction}

So far we have removed the generators that are not required anymore. It remains to adjust the lengths of the remaining generators so that the spanned zonotope under-approximates or over-approximates the Minkowski difference $\mathcal{Z}_m\ominus \mathcal{Z}_s$. It is obvious that a polytope is a subset of another polytope with identical $\hat{C}$ matrix, if and only if the $\hat{d}$ vector is smaller in each dimension:
\begin{subequations} \label{eq:polytopeEnclosure}
\begin{align}
  & \{x | \hat{C}x \leq \hat{d}'\} \subseteq \{x | \hat{C}x \leq \hat{d}\} \\
  \Longleftrightarrow & \forall i: \hat{d}_i' \leq \hat{d}_i \label{eq:subsetCondition} \\
  \implies & \bigg \{x \bigg| \begin{bmatrix} \hat{C} \\ C^{red} \end{bmatrix} x \leq \begin{bmatrix} \hat{d}' \\ d^{red}\end{bmatrix} \bigg\} \subseteq \{x | \hat{C}x \leq \hat{d}\}. \label{eq:redundantHalfspaces}
\end{align}
\end{subequations}
This result is used to compute the following under-approximation and over-approximation by stretching the generators of the minuend using the vector of stretching factors $\mu = \begin{bmatrix} \mu_1, & \ldots, & \mu_{p_m} \end{bmatrix}$. To apply \eqref{eq:polytopeEnclosure}, we first extract the $d^+$ vector of the Minkowski difference from Theorem~\ref{thm:hMinkowskiDifference}:
 \begin{equation} \label{eq:dPlusOne}
 \begin{split}
  d^+ =& C^+\, (c\^m -c\^s) + \underbrace{\sum_{\upsilon=1}^{p_m} |C^+\, g\^{m,\upsilon}|}_{=: \Delta d} - \underbrace{\sum_{\upsilon=1}^{p_s} |C^+\, g\^{s,\upsilon}|}_{=:\Delta d_{trans}}.
 \end{split}
 \end{equation}
The $d^+(\mu)$ values of the halfspace representation of $\mathcal{Z}_m$ for a given stretching vector $\mu$ are (see Theorem~\ref{thm:hZonotope}):
\begin{equation} \label{eq:dPlusTwo}
 \begin{split}
  d^+(\mu) =& C^+\, (c\^m -c\^s) + \sum_{\upsilon=1}^{\hat{p}_m} |C^+\, \hat{g}\^{m,\upsilon}| \mu_\upsilon = C^+\, (c\^m -c\^s) + |C^+\, \hat{G}\^{m}| \mu.
 \end{split}
\end{equation}
Now, we can enforce that $d^+(\mu)$ of the stretched minuend  is smaller than $d^+$ using the following theorem.

\begin{theorem}[Under-Approximative G-Representation of the Minkowski Difference] \label{thm:gMinkowskiDifference_under}
Given are the two zonotopes $\mathcal{Z}_m = (c\^m,G\^m)$ and $\mathcal{Z}_s = (c\^s,G\^s)$. An under-approximative generator representation $(\hat{c},\hat{G})$ of $\mathcal{Z}_m \ominus \mathcal{Z}_s$ is:
\begin{equation*}
 \begin{split}
  \hat{c} &= c\^m - c\^s, \quad 
  \hat{G} = \hat{G}\^m \mu^*, 
 \end{split}
\end{equation*}
where $\mu^*$ is the solution of the linear program
\begin{equation}  \label{eq:linProgram_under}
\begin{split}
 & \mu^* = \underset{\mu}{\mathtt{argmax}} \quad \zeta^T \mu \\
 \text{subject to } & \hat{A} \mu \leq \hat{b} \\
 & \mu \geq \mathbf{0},
\end{split}
\end{equation}
in which $\zeta$ is a user-specified vector, $\hat{A} = |C^+\, \hat{G}\^{m}|$, and $\hat{b} = \Delta d - \Delta d_{trans}$. The values of $C^+$, $\Delta d$, and $\Delta {d}_{trans}$ are computed as in Theorem~\ref{thm:hMinkowskiDifference}. The result also holds when removing redundant halfspaces.
\end{theorem}

\begin{proof}
 The normal vectors of the Minkowski difference are those of the minuend (see Theorem~\ref{thm:hMinkowskiDifference}). Thus, it suffices according to \eqref{eq:polytopeEnclosure} that the $d^+(\mu)$ vector of the stretched minuend in \eqref{eq:dPlusTwo} is smaller than $d^+$ of the exact solution in \eqref{eq:dPlusOne} to obtain an under-approximation:
\begin{equation} 
 \begin{split}
  C^+\, (c\^m -c\^s) + |C^+\, \hat{G}\^{m}| \mu &\leq C^+\, (c\^m -c\^s) + \Delta d - \Delta d_{trans} \\
  \Leftrightarrow \underbrace{|C^+\, \hat{G}\^{m}|}_{\hat{A}} \mu &\leq \underbrace{\Delta d - \Delta d_{trans}}_{\hat{b}}.
 \end{split}
\end{equation}
In order to obtain large positive values of $\mu$ and thus a tight under-approximation, we use the linear program in \eqref{eq:linProgram_under}.
\end{proof}

The vector $\zeta$ can be freely chosen. As a possible heuristic, we use $\zeta_i=\|g\^{m,i}\|_2$ so that the stretching factors of long generators receive a larger weight.

As shown in \eqref{eq:redundantHalfspaces}, the under-approximation in Theorem~\ref{thm:gMinkowskiDifference_under} also holds if redundant halfspaces specified by $C^{red}$ and $d^{red}$ are not removed. However, \eqref{eq:redundantHalfspaces} is only applicable for under-approximations so that for over-approximations, we can only use \eqref{eq:subsetCondition}, which implies that the halfspaces of the minuend and the Minkowski difference have to be equal. Thus, one can only remove those halfspaces that are not generated by Theorem~\ref{thm:hMinkowskiDifference} after removing generators according to Sec.~\ref{sec:generatorRemoval}. 
The over-approximation is obtained by computing the smallest values of $\mu$, such that $\hat{A} \mu \geq \hat{b}$ (size of $\hat{A}$ and $\hat{b}$ may differ from under-approximative Minkowski difference due to a different number of removed halfspaces). To compensate for the effect that less halfspaces can be removed compared to the under-approximation, we move the redundant halfspaces as much as possible inwards so that they touch the exact Minkowski difference, as shown next. 


\begin{theorem}[Over-Approximative G-Representation of the Minkowski Difference] \label{thm:gMinkowskiDifference_over}
Given are the two zonotopes $\mathcal{Z}_m = (c\^m,G\^m)$ and $\mathcal{Z}_s = (c\^s,G\^s)$. An over-approximative generator representation $(\tilde{c},\tilde{G})$ of $\mathcal{Z}_m \ominus \mathcal{Z}_s$ is:
\begin{equation*} 
 \begin{split}
  \tilde{c} &= c\^m - c\^s, \quad 
  \tilde{G} = \hat{G}\^m \mu^*, 
 \end{split}
\end{equation*}
where $\mu^*$ is the solution of the linear program
\begin{equation} \label{eq:linProgram_over}
\begin{split}
 & \mu^* = \underset{\mu}{\mathtt{argmax}} \quad -\zeta^T \mu \\
 \text{subject to } & - \hat{A} \mu \leq - \tilde{b} \\
 & \mu \geq \mathbf{0},
\end{split}
\end{equation}
in which $\zeta$ is a user-specified vector, $\hat{A} = |C^+\, \hat{G}\^{m}|$ (identical to Theorem~\ref{thm:gMinkowskiDifference_under}), and $\tilde{b}$ is obtained by another linear program:
\begin{equation} \label{eq:shrinkPolytope}
\begin{split}
 & \tilde{b}_i = \underset{\mu}{\mathtt{argmax}} \quad C^+_i \mu \\
 \text{subject to } & C^+ \mu \leq \Delta d - \Delta d_{trans} \\
 & \mu \geq \mathbf{0},
\end{split}
\end{equation}
where the values of $C^+$, $\Delta d$, and $\Delta d_{trans}$ are computed as in Theorem~\ref{thm:hMinkowskiDifference}.
\end{theorem}

\begin{proof}
 The structure of the proof is similar to that of Theorem~\ref{thm:gMinkowskiDifference_under}. Because the normal vectors of the Minkowski difference are those of the minuend (see Theorem~\ref{thm:hMinkowskiDifference}), it suffices according to \eqref{eq:polytopeEnclosure} that the $d^+(\mu)$ vector of the stretched minuend in \eqref{eq:dPlusTwo} is greater than $d$ in Theorem~\ref{thm:hMinkowskiDifference}:
\begin{equation} 
 \begin{split}
  C^+\, (c\^m -c\^s) + |C^+\, \hat{G}\^{m}| \mu &\geq C^+\, (c\^m -c\^s) + \Delta d - \Delta d_{trans} \\
  \Leftrightarrow \underbrace{|C^+\, \hat{G}\^{m}|}_{\hat{A}} \mu &\geq \underbrace{\Delta d - \Delta d_{trans}}_{\hat{b}}.
 \end{split}
\end{equation}
In order to obtain small values of $\mu$ and thus a tight over-approximation, we use the linear program in \eqref{eq:linProgram_over}.

As previously mentioned, we require that the number of halfspace of the minuend is equal to that of the Minkowski difference to apply \eqref{eq:polytopeEnclosure}. Thus, we keep all halfspaces, but reduce the conservatism by finding $\tilde{b} \leq \hat{b}$ (comparison is performed elementwise) through the linear program in \eqref{eq:shrinkPolytope}. This linear program determines the furthest one can travel in the direction of the corresponding halfspace without leaving the minuend, which returns how close a redundant halfspace can be moved to the minuend; non-redundant halfspace are obviously unaffected.
\end{proof}
The tightness of the under-approximation and over-approximation is evaluated in Sec.~\ref{sec:numericalExperiments}.



\subsection{Reducing the Computational Complexity} \label{sec:reductionMethod}

In order to reduce the number of required halfspaces, one can over-approximate the minuend to obtain an over-approximative Minkowski difference \cite{Kopetzki2017, Yang2018}. However, to under-approximate the minuend, no similarly powerful methods compared to over-approximative methods exist. To address this problem, we present a novel reduction technique for the under-approximation of the Minkowski difference. Instead of reducing the number of generators of the minuend before the computation of the Minkowski difference as proposed for the over-approximation, we split the minuend into a part that encloses the subtrahend and a remainder. Then the Minkowski difference is computed without the remainder and the remainder is later added via Minkowski sum. To show that this approach is under-approximative, we first provide the following intermediate result.

\begin{proposition} \label{prop:setEnclosure}
 Given that $\mathcal{C} \subseteq \mathcal{A}$, the following enclosure holds:
 \begin{equation*}
  \{a+b | a \oplus \mathcal{C} \subseteq \mathcal{A}, b \in \mathcal{B}\} \subseteq \{a + b| a+b \oplus \mathcal{C} \subseteq \mathcal{A} \oplus \mathcal{B}\}.
 \end{equation*}
\end{proposition}
\begin{proof}
 Since $\mathcal{C} \subseteq \mathcal{A}$, for all $b\in\mathcal{B}$ there exists an $a$ so that
 \begin{equation*}
 \begin{split}
   a \oplus \mathcal{C} & \subseteq \mathcal{A} \\
  \Leftrightarrow  a + b \oplus \mathcal{C} & \subseteq \{\tilde{a} + b| \tilde{a} \in \mathcal{A}, b \in \mathcal{B}\} \\
  \Rightarrow  a + b \oplus \mathcal{C} &\subseteq \{\tilde{a} + \tilde{b}| \tilde{a} \in \mathcal{A}, \tilde{b} \in \mathcal{B}\} \\
  \Leftrightarrow  a + b \oplus \mathcal{C} & \subseteq \mathcal{A} \oplus \mathcal{B}
 \end{split}
 \end{equation*}
\end{proof}

As previously mentioned, we split the minuend into two zonotopes $\mathcal{Z}_m = \mathcal{Z}_{m,encl} \oplus \mathcal{Z}_{m, rem}$ so that $\mathcal{Z}_{m,encl} \supseteq \mathcal{Z}_s$. This is only possible if for some $\xi$, $\mathcal{Z}_m \supseteq \mathcal{Z}_s + \xi$. Otherwise, $\mathcal{Z}_m \ominus \mathcal{Z}_s$ is empty.

\begin{proposition} 
 For $\mathcal{Z}_{m,encl} \supseteq \mathcal{Z}_s$ it holds that
 \begin{equation*}
  (\mathcal{Z}_{m,encl} \ominus \mathcal{Z}_s) \oplus \mathcal{Z}_{m, rem} \subseteq (\mathcal{Z}_{m,encl} \oplus \mathcal{Z}_{m, rem}) \ominus \mathcal{Z}_s
 \end{equation*}
\end{proposition}
\begin{proof}
 Starting with the definition of the Minkowski difference, we have the following relations:
 \begin{equation*} 
 \begin{split}
  & (\mathcal{Z}_{m,encl} \oplus \mathcal{Z}_{m, rem}) \ominus \mathcal{Z}_s \\
  \overset{\eqref{eq:MinkowskiDifference}}{=}&  \{x \in \mathbb{R}^n | x \oplus \mathcal{Z}_s \subseteq (\mathcal{Z}_{m,encl} \oplus \mathcal{Z}_{m, rem})\} \\
  = & \{(x_{encl} + x_{rem}) \in \mathbb{R}^n | (x_{encl} + x_{rem}) \oplus \mathcal{Z}_s \subseteq (\mathcal{Z}_{m,encl} \oplus \mathcal{Z}_{m, rem})\} \\
  \overset{Prop.~\ref{prop:setEnclosure}}{\supseteq} & \{(x_{encl} + x_{rem}) \in \mathbb{R}^n | x_{encl} \oplus \mathcal{Z}_s \subseteq \mathcal{Z}_{m,encl}, \,\, x_{rem} \in \mathcal{Z}_{m, rem}\} \\
  = & \{x_{encl} \in \mathbb{R}^n | x_{encl} \oplus \mathcal{Z}_s \subseteq \mathcal{Z}_{m,encl}\} \oplus \mathcal{Z}_{m, rem} \\
  \overset{\eqref{eq:MinkowskiDifference}}{=} & (\mathcal{Z}_{m,encl} \ominus \mathcal{Z}_s) \oplus \mathcal{Z}_{m, rem}
 \end{split}
\end{equation*}
\end{proof}

In order to use as few generators as possible to form $\mathcal{Z}_{m,encl}$ so that $\mathcal{Z}_s$ is enclosed, we first sort the generators of $\mathcal{Z}_{m}$ by their length using the operator $\mathrm{sort}(\cdot)$ in Alg.~\ref{alg:generatorSelection}. To avoid checking the enclosure of $\mathcal{Z}_s$ too often, we compare the size of the interval hulls (i.e., the enclosing axis-aligned boxes) of $\mathcal{Z}_{m,encl}$ and $\mathcal{Z}_s$. The operator for returning interval hulls is $\mathrm{box}(\cdot)$ and the radius $r\in\mathbb{R}^n$ of the interval hulls is defined as $[c-r,c+r]$ and returned by $\mathrm{rad}(\cdot)$. As soon as the worst ratio $\max_i \frac{\mathrm{rad}(\mathrm{box}(\mathcal{Z}_s)_i)}{\mathrm{rad}(\mathrm{box}(\mathcal{Z}_{m,encl})_i)}$ is below a user-defined threshold $\underline{\mu}$, no more generators are added to $\mathcal{Z}_{m,encl}$. Afterwards, more generators are added in case $\mathcal{Z}_{m,encl} \not\supseteq \mathcal{Z}_s$. This can be done by binary search or simply by adding single generators as shown in Alg.~\ref{alg:generatorSelection}.

\begin{algorithm}
\caption{Generator Selection for $\mathcal{Z}_{m}$} \label{alg:generatorSelection}
\begin{algorithmic}[1]
	\Require $\mathcal{Z}_{m}=(c_m, G_m)$, $\mathcal{Z}_{s}=(c_s, G_s)$.
	\Ensure $\mathcal{Z}_{m,encl}$
	\State $G_{sort} \gets \mathrm{sort}(G_{m})$
	\State $r_{sel} \gets 0$
	\State $r_s \gets \mathrm{rad}(\mathrm{box}(\mathcal{Z}_{s}))$
	\State $\mu \gets \mathrm{inf}$
	\State $\xi \gets 0$
    \While{$\mu > \underline{\mu} \land \xi < n_m$}
        \State $\xi \gets \xi + 1$
        \State $r_{sel} \gets r_{sel} + |g_{sort}\^\xi|$
        \State $\mu \gets \max_i r_{s,i}/r_{sel,i}$
    \EndWhile
    \While{$\mathcal{Z}_s \supset (c_s, g_{sort}\^1, \ldots, g_{sort}\^\xi) \land \xi < n_m$}
        \State $\xi \gets \xi + 1$
    \EndWhile
    \State $\mathcal{Z}_{m,encl} \gets (c_m, g_{sort}\^1, \ldots, g_{sort}\^\xi)$
\end{algorithmic}
\end{algorithm}

Before evaluating the proposed approaches for under-approximating and over-approximating the Minkowski difference, we present special cases with exact solutions.

\section{Exact Solutions} \label{sec:exactSolution}

In this section, we present two special cases with exact solutions: zonotopes with aligned generators and two-dimensional zonotopes. Let us first present the case with aligned generators.

\begin{definition}[Aligned zonotopes]
 Given are two zonotopes $\mathcal{Z}_1=(c, g\^{1}, \ldots, g\^{p_1})$ and $\mathcal{Z}_2=(d, h\^{1}, \ldots, h\^{p_2})$. Without loss of generality, we assume that the generators within each of these zonotopes are not aligned, i.e., for $\mathcal{Z}_1$ it is not true that $\exists i, \, \exists j\neq i, \, \exists \gamma: \gamma \, g\^i = g\^j$, where $i,j \in \mathbb{N}$ and $\gamma \in \mathbb{R}$; this applies analogously to $\mathcal{Z}_2$. Otherwise, as mentioned above, one could adjust $g\^i := (\gamma + 1) g\^i$ and remove $g\^j$. If 
 \begin{equation*}
  \forall i \in \{1, \ldots p_2\} \, \exists j \in \{1, \ldots p_1\} \, \exists \gamma \in \mathbb{R}: \gamma \, g\^j = h\^i 
 \end{equation*}
 we say that the zonotope $\mathcal{Z}_1$ is aligned with the zonotope $\mathcal{Z}_2$. 
\end{definition}
If $\mathcal{Z}_1$ is aligned with $\mathcal{Z}_2$ and vice versa, we say that the zonotopes are mutually aligned.

\begin{proposition}[Minkowski difference for aligned generators]
 Given is a zonotope $\mathcal{Z}_m$ which is aligned with $\mathcal{Z}_s$. Without loss of generality, we re-order the generators so that $\forall i \in \{1, \ldots p_s\}$ we have that $\gamma_i g\^{m,i} = g\^{s,i}$, $\gamma_i \in \mathbb{R}$. The Minkowski difference can be computed exactly:
 \begin{equation}\label{eq:alignedDifference}
  \mathcal{Z}_m \ominus \mathcal{Z}_s = (c\^m - c\^s, g\^{m,1} - g\^{s,1}, \ldots, g\^{m,p_s} - g\^{s,p_s}, g\^{m,p_s+1}, \ldots, g\^{m,p_m}).
 \end{equation}
\end{proposition}
\begin{proof}
 The proof is straightforward. The Minkowski sum of $\mathcal{Z}_m$ and $\mathcal{Z}_s$ is obviously 
 \begin{equation}\label{eq:alignedSum}
  \mathcal{Z}_m \oplus \mathcal{Z}_s = (c\^m + c\^s, g\^{m,1} + g\^{s,1}, \ldots, g\^{m,p_s} + g\^{s,p_s}, g\^{m,p_s+1}, \ldots, g\^{m,p_m}).
 \end{equation}
 Thus, combining \eqref{eq:alignedDifference} and \eqref{eq:alignedSum} results in $(\mathcal{Z}_m \ominus \mathcal{Z}_s) \oplus \mathcal{Z}_s = \mathcal{Z}_m$, which shows that the Minkowski difference for a zonotope aligned with another one is exact.
\end{proof}

Next, we show that Minkowski difference is exact in the two-dimensional case. We assume that the redundant generators have been removed as presented in Sec.~\ref{sec:generatorRemoval}.

\begin{proposition}[Minkowski difference of two-dimensional zonotopes] \label{prop:two-dimensionalCase}
 The Minkowski difference of two-dimensional zonotopes is the zonotope
 \begin{equation*}
  \big(c\^m -c\^s, \hat{G}\^m |\hat{C}^+ \hat{G}\^{m}|^{-1} (\Delta \hat{d} - \Delta \hat{d}_{trans})\big).
 \end{equation*}
\end{proposition}
\begin{proof}
 To show that the result is exact, the $d$ vector of the Minkowski difference has to be exact (see Theorem~\ref{thm:hMinkowskiDifference}). After equating \eqref{eq:dPlusOne} and \eqref{eq:dPlusTwo} using non-redundant halfspaces and the reduced set of generators, we obtain 
 \begin{equation} \label{eq:equal_d}
 \begin{split}
  \hat{C}^+\, (c\^m -c\^s) + |\hat{C}^+\, \hat{G}\^{m}| \mu =& \hat{C}^+\, (c\^m -c\^s) + \Delta \hat{d} - \Delta \hat{d}_{trans} \\
  |\hat{C}^+\, \hat{G}\^{m}| \mu =& \Delta \hat{d} - \Delta \hat{d}_{trans}.
  \end{split}
  \end{equation}
 The above equation can be solved exactly since each generator corresponds to one normal vector in two dimensions. This can be easily shown by evaluating Theorem~\ref{thm:hZonotope} for the two-dimensional case, resulting in the $\th{i}$ normal vector for the $\th{i}$ generator: 
 \begin{equation} \label{eq:normalVector_2D}
  C_i^+= \frac{\begin{bmatrix} g_2\^i & g_1\^i \end{bmatrix}}{\left\|\begin{bmatrix} g_2\^i & g_1\^i \end{bmatrix}^T\right\|_2}.
 \end{equation}
 Thus, \eqref{eq:equal_d} has a unique solution so that the matrix $|\hat{C}^+\, \hat{G}\^{m}|$ in \eqref{eq:equal_d} can be inverted to obtain the exact stretching vector 
 \begin{equation*}
  \mu = |\hat{C}^+\, \hat{G}\^{m}|^{-1}(\Delta \hat{d} - \Delta \hat{d}_{trans}).
 \end{equation*}
 After stretching the generators of the minuend and shifting the center, we obtain the result of the proposition.
\end{proof}
Subsequently we evaluate our novel approaches using randomly generated zonotopes. Obviously, for this type of evaluation, aligned zonotopes will almost surely not be generated.

\section{Numerical Experiments} \label{sec:numericalExperiments}

We assess the performance of computing the Minkowski difference for random zonotopes of various orders and dimensions. Computation times are compared with different methods for computing exact, under-approximative, and over-approximative results. The accuracy is evaluated by the accuracy measure $\Theta$ which is based on the ratio of the volume of two sets $\mathcal{S}$ compared to a reference set $\mathcal{S}_{ref}$. Let us also introduce the operator $\mathrm{volume()}$ returning the volume of a set. To compare the accuracy obtained for different dimensions, the $n^{th}$ root of the volume ratio is taken, which is equivalent to the ratio of the edge length of n-dimensional cubes containing the corresponding volumes. Combining the mentioned aspects, the accuracy measure is computed as
\begin{equation} \label{eq:accuracyMeasure}
	\Theta=\left(\frac{\mathrm{volume}(\mathcal{S})}{\mathrm{volume}(\mathcal{S}_{ref})}\right)^{\frac{1}{n}}.
\end{equation}
In this work, we use the set obtained from the under-approximative method HC-u (explained later in detail) as a reference set. 

All computations are performed on a laptop with an Intel Core i7-8565U CPU with 1.80 GHz and 24 GB of memory. Parallelization of algorithms is not used in all evaluations. 

\subsection{Random Generation of Zonotopes}

To fairly assess the performance, random zonotopes are generated. Each scenario has the following parameters:
\begin{itemize}
 \item The dimension $n$ of the Euclidean space $\mathbb{R}^n$.
 \item The order $\varrho_s$ of the subtrahend $\mathcal{Z}_s$.
 \item The order $\varrho_m$ of the minuend $\mathcal{Z}_m$.
 \item The maximum length $l^{s,max}$ of the generators $\|g\^{s,i}\|_2$ of $\mathcal{Z}_s$ is selected as $1$ without loss of generality.
 \item The maximum length $l^{m,max}$ of the generators $\|g\^{m,i}\|_2$ of $\mathcal{Z}_m$ is selected as $10\frac{\varrho_s}{\varrho_m}$. The fraction $\frac{\varrho_s}{\varrho_m}$ ensures that the size of the minuend and subtrahend are comparable when using different orders. The factor of $10$ is selected to ensure that most results are not empty. 
 \item The number of randomly generate instances for a given combination of the above parameters is chosen as $100$. All provided results are averaged over these instances.
\end{itemize}
Given the above parameters, a simple method for obtaining random generators would be to first randomly generate each entry of a generator by uniformly sampling values within $[-1,1]$. The generator would then be stretched to a length that is uniformly distributed within $[0,l^{max}]$. However, the directions of the resulting generators would not be uniformly distributed. Thus, we first generate points that are uniformly distributed on a unit hypersphere according to \cite{Muller1959}. Next, the generators are defined as the vector from the origin to the points on the hypersphere, which are stretched to achieve the desired length of the generators $l\^i=\|g\^{i}\|_2$. This is uniformly distributed within $[0,l^{max}]$. 

\subsection{Comparison with Polytope Implementation}\label{sec:comparisonPolytope}

In this subsection, we compare the computation times of our own MATLAB implementations with a method developed for polytopes. In particular, the following methods are evaluated:
\begin{itemize}
 \item Halfspace conversion (HC): This method returns the exact polytope of the Minkowski difference using the halfspace conversion as presented in Theorem~\ref{thm:hMinkowskiDifference}.
 \item Under-approximative enclosure of halfspace conversion (HC-u): This method returns a zonotope under-approximating the result of the halfspace conversion (HC) as presented in Theorem~\ref{thm:gMinkowskiDifference_under}. This method does not reduce the number of generators upfront (see Sec.~\ref{sec:generatorRemoval}) to save computation time. 
 \item Polytope conversion (Polytope): This method first converts the zonotopes to polytopes using Theorem~\ref{thm:hZonotope} (this conversion is excluded from measuring the computation time). The Minkowski difference is then computed with the \textit{Multi-Parametric Toolbox 3.0} using its default settings \cite{Herceg2013}.
\end{itemize}
We additionally assess the computation time when combining Minkowski difference with Minkowski addition, as it is required in many applications, see e.g. \cite{Rakovic2006}. In this evaluation, we compute $(\mathcal{Z}_m \ominus \mathcal{Z}_s)\oplus \mathcal{Z}_s$---this is not meaningful, but the average computation times would not change if we add a set other than $\mathcal{Z}_s$ with the same order and dimension. 

Tab.~\ref{tab:results} lists the average computation times. Since it is computationally infeasible to compute the volume of the obtained polytopes for dimensions greater than or equal to four, we cannot provide the accuracy of the results as defined in \eqref{eq:accuracyMeasure}. However, one can see later in Tab.~\ref{tab:results_over} that the values of the accuracy $\Theta$ for the under-approximation and the over-approximation are fairly close.

\begin{table*}[h!tb]
\centering
\caption{Average computational times of the Minkowski difference and its combination with the Minkowski addition for various scenarios. Each scenario is run 100 times and we declare a scenario as \textit{did not finish} (\textrm{dnf}) when a single set computation takes more than two hours (amounting to more than one week for $100$ instances) or the memory is exceeded. \mbox{} \vspace{0.3cm} \mbox{}}
\label{tab:results}
\centering
\begin{tabular}{@{}rrrcrrrcrrr@{}}\toprule
&&&& \multicolumn{5}{c}{Computation Time [s]} \\\cmidrule{5-9}
Dim. & \multicolumn{2}{c}{Order} && \multicolumn{3}{c}{$\mathcal{Z}_m \ominus \mathcal{Z}_s$} && \multicolumn{3}{c}{$(\mathcal{Z}_m \ominus \mathcal{Z}_s)\oplus \mathcal{Z}_s$} \\
\cmidrule{2-3} \cmidrule{5-7} \cmidrule{9-11}
$n$           & $\mathcal{Z}_m$         & $\mathcal{Z}_s$ && HC & HC-u & Polytope && HC & HC-u & Polytope \\ \midrule
$2$ & $2$ & $2$ && $0.001$ & $0.004$ & $0.006$ && $0.014$ & $0.004$ & $0.026$  \\  
$2$ & $4$ & $2$ && $0.001$ & $0.005$  & $0.006$  && $0.018$ & $0.005$ & $0.054$  \\  
$2$ & $2$ & $4$ && $0.001$ & $0.004$  & $0.009$  && $0.011$ & $0.004$ & $0.060$  \\
$2$ & $4$ & $4$ && $0.001$ & $0.004$  & $0.010$  && $0.015$ & $0.004$ & $0.158$  \\ \hline
$4$ & $2$ & $2$ && $0.001$ & $0.004$  & $0.148$  && \textrm{dnf} & $0.004$ & \textrm{dnf}  \\  
$4$ & $4$ & $2$ && $0.004$  & $0.013$ & $0.419$ && \textrm{dnf} & $0.013$ & \textrm{dnf} \\  
$4$ & $2$ & $4$ && $0.001$  & $0.005$ & $44.19$ && \textrm{dnf} & $0.005$ & \textrm{dnf} \\
$4$ & $4$ & $4$ && $0.004$  & $0.014$ & $91.96$ && \textrm{dnf} & $0.014$ & \textrm{dnf} \\ \hline
$6$ & $2$ & $2$ && $0.008$  & $0.017$  & $547.7$ && \textrm{dnf} & $0.017$ & \textrm{dnf} \\ 
$6$ & $4$ & $2$ && $0.441$  & $0.955$  & \textrm{dnf} && \textrm{dnf} & $0.955$ & \textrm{dnf} \\ 
$6$ & $2$ & $4$ && $0.012$  & $0.021$ & \textrm{dnf} && \textrm{dnf} & $0.021$ & \textrm{dnf} \\ 
$6$ & $4$ & $4$ && $0.473$  & $1.348$ & \textrm{dnf} && \textrm{dnf} & $1.348$ & \textrm{dnf} \\ 
\bottomrule
\end{tabular}
\end{table*}

Tab.~\ref{tab:results} shows that the zonotope implementations are faster than the polytope implementation of the \textit{Multi-Parametric Toolbox 3.0}, e.g. for $6$ dimensions and an order of $2$ for the minuend and the subtrahend, the method HC is more than $6\cdot10^4$ times faster. The difference in computation time is even more apparent when combining Minkowski difference and Minkowski addition. For dimensions equal or larger than $4$, we obtain the result \textit{did not finish} ($\mathtt{dnf}$) when the Minkowski difference is represented as a polytope. Thus, the method HC-u returning zonotopes is several orders of magnitude faster when combining Minkowski difference with Minkowski addition.

\subsection{Comparison of Methods for Under-Approximation}

As a further comparison, we evaluate our under-approximative approach with and without order reduction and compare those results with an alternative method presented in \cite[Theorem~7]{Raghuraman2022}. This alternative method is based on an efficient enclosure check of zonotopes presented in \cite[Corollary 4]{Sadraddini2019}. Since we will present an efficient heuristic for this method, we briefly recall it and refer the reader to \cite[Theorem~7]{Raghuraman2022} for the proof.
Let us introduce $\tilde{\Gamma}\in \mathbb{R}^{n_m \times (n_m + 2n_s)}$ and $\tilde{\beta} \in \mathbb{R}^{n_m}$. An under-approximation $\mathcal{Z}_d = (c\^d,[G\^m G\^s]\mathrm{diag}(\phi))$, $\phi \in \mathbb{R}_+^{n_m+n_s}$ can be obtained according to \cite[Theorem~7]{Raghuraman2022} if a solution to the following linear program exists:
\begin{equation} \label{eq:minkowskiDifferenceProblem}
\begin{split}
  \begin{bmatrix} [G\^m, \, G\^s] \mathrm{diag}(\phi) & G\^s \end{bmatrix} &= G\^m \tilde{\Gamma}, \\
  c\^m -  (c\^d + c\^s) &= G\^m \tilde{\beta}, \\
  \|[\tilde{\Gamma}, \tilde{\beta}]\|_\infty &\leq 1.
\end{split}
\end{equation}
A disadvantage of this method is that the order of the result is $\varrho_m + \varrho_s$, i.e., that of the minuend plus the subtrahend, while the maximum order of the exact result is $\varrho_m$ as shown in Theorem~\ref{thm:hMinkowskiDifference}. It is also worth mentioning that the approach in \cite[Theorem~7]{Raghuraman2022} cannot be modified in a straightforward way to compute over-approximations because it is based on \cite[Corollary 4]{Sadraddini2019} for conservatively checking the enclosure of zonotopes.

We compare the following methods for under-approximating the Minkowski difference:
\begin{itemize}
 \item Method by Raghuraman and Koeln (RK): This method implements \eqref{eq:minkowskiDifferenceProblem}. In \cite{Raghuraman2022} the cost function for $\phi$ is not provided. As a standard choice, we minimize $\sum_i \phi_i$. To minimize computational overhead, we did not use tools such as YALMIP to solve \eqref{eq:minkowskiDifferenceProblem}, but directly provided \eqref{eq:minkowskiDifferenceProblem} in the form required by the MATLAB function \texttt{linprog}.
 \item Method by Raghuraman and Koeln with generator length priority (RK-p): This method is identical to RK, with a small but crucial difference: We use the cost function $\sum_i \|[G\^m, \, G\^s]_{(:,i)}\|_2 \phi_i$, where $[G\^m, \, G\^s]_{(:,i)}$ returns the $\th{i}$ column vector of the concatenated matrix consisting of $G\^m$ and $G\^s$. With other words, the importance of $\phi_i$ is weighted by the length of the generator it is multiplied with in \eqref{eq:minkowskiDifferenceProblem}.
 \item Under-approximative enclosure of halfspace conversion (HC-u): See description in Sec.~\ref{sec:comparisonPolytope}.
 \item Under-approximative enclosure of halfspace conversion using zonotope reduction (HC-r): This method is identical to HC-u with the added reduction from Sec.~\ref{sec:reductionMethod}. The parameter $\underline{\mu}$ is chosen as $0.3$. 
\end{itemize}

Tab.~\ref{tab:results_under} lists the average computation times and the accuracy measure $\Theta$. Since it is computationally infeasible to compute the volume of the obtained zonotopes for dimensions greater than or equal to six, we only provided results up to six dimensions. For six dimensions we could only present the results for a reduced amount of generators due to reaching the memory limit when computing the zonotope volumes. 

\begin{table*}[h!tb]
\centering
\caption{Average computational time of under-approximating the Minkowski difference.}
\label{tab:results_under}
\centering
\begin{tabular}{@{}rrrcrrrrcrrrr@{}}\toprule
Dim. & \multicolumn{2}{c}{Order} && \multicolumn{4}{c}{Accuracy $\Theta$} && \multicolumn{4}{c}{Computation Time [s]} \\
\cmidrule{2-3} \cmidrule{5-8} \cmidrule{10-13}       
$n$           & $\mathcal{Z}_m$         & $\mathcal{Z}_s$ && RK & RK-p & HC-u & HC-r  && RK & RK-p & HC-u & HC-r \\ \midrule
$2$ & $10$ & $10$ && $0.443$  & $1.000$ & $1.000$  & $0.991$ && $0.117$  & $0.116$ & $0.005$  & $0.010$ \\
$2$ & $20$ & $10$ && $0.474$  & $1.000$ & $1.000$  & $0.996$ && $0.007$  & $0.639$ & $0.637$  & $0.012$ \\
$2$ & $10$ & $20$ && $0.458$  & $1.000$ & $1.000$  & $0.989$ && $0.005$  & $0.294$ & $0.293$  & $0.013$ \\
$2$ & $20$ & $20$ && $0.406$  & $1.000$ & $1.000$  & $0.996$ && $1.330$  & $1.320$ & $0.007$  & $0.019$ \\ \hline 
$4$ & $10$ & $10$ && $0.162$  & $0.992$ & $1.000$  & $0.969$ && $1.894$  & $1.899$ & $0.265$  & $0.039$ \\
$4$ & $20$ & $10$ && $0.158$  & $0.995$ & $1.000$  & $0.987$ && $8.570$  & $8.548$ & $4.756$  & $0.074$ \\
$4$ & $10$ & $20$ && $0.123$  & $0.991$ & $1.000$  & $0.962$ && $3.648$  & $3.645$ & $0.303$  & $0.074$ \\
$4$ & $20$ & $20$ && $0.132$  & $0.994$ & $1.000$  & $0.986$ && $19.18$  & $19.06$ & $4.682$  & $0.198$ \\ \hline  
$6$ & $8$ & $8$ && $0.011$ & $0.979$ & $1.000$ & $0.902$ && $2.770$  & $2.717$  & $89.10$ & $0.077$ \\ 
\bottomrule
\end{tabular}
\end{table*}

%
%
%
%
%

The results in Tab.~\ref{tab:results_under} show that the effect of the newly proposed generator weighting for the approach in \cite[Theorem~7]{Raghuraman2022} is enormous. While the accuracy for RK decreases consistently with the dimension down to only $\Theta = 0.011$ for six dimensions, the accuracy stays high for \mbox{RK-p} resulting in $\Theta = 0.979$ for the same number of dimensions. Another interesting result is that the order reduction results in significantly shorter computation times. For six-dimensional problems, HC-r was more than three orders of magnitude faster on average compared to HC-u. However, for small-dimensional problems, HC-u is faster than HC-r due to avoiding the computational overhead for the order reduction. In terms of overall accuracy, HC-u is more accurate than RK-p, which in turn is more accurate than HC-r across all instances. Please note, however, that by tuning $\underline{\mu}$, HC-r can be made faster or more accurate. Regarding computation times, HC-r provided the shortest times while providing a good accuracy of $\Theta>0.9$ for all instances. 

\subsection{Comparison of Methods for Over-Approximation} \label{sec:comparison-over}

In this subsection, we compare methods for computing zonotopes that over-approximate the Minkowski difference of two zonotopes. For this problem no alternative methods exist. Because computing the exact volume is already infeasible in four-dimensional space, we compare the results with the previously introduced under-approximation HC-u. In particular, we compare the following methods:
\begin{itemize}
 \item Under-approximative enclosure of halfspace conversion (HC-u): See description in Sec.~\ref{sec:comparisonPolytope}.
 \item Over-approximative enclosure of halfspace conversion (HC-o): This method returns a zonotope over-approximating the result of the halfspace conversion (HC) as presented in Theorem~\ref{thm:gMinkowskiDifference_over}. This method does not reduce the number of generators upfront (see Sec.~\ref{sec:generatorRemoval}) to save computation time. 
 \item Coarse over-approximative enclosure of halfspace conversion (HC-co): This method is identical to HC-o, except that the shrinking of the halfspaces as presented in \eqref{eq:shrinkPolytope} is removed to save computational time.
\end{itemize}

The average computation times and the accuracy measure $\Theta$ are listed in Tab.~\ref{tab:results_under}. As explained above, the numerical experiments are limited to six dimensions due to otherwise reaching the memory limit when computing the zonotope volumes. 

\begin{table*}[h!tb]
\centering
\caption{Average computational time of under-approximating the Minkowski difference.}
\label{tab:results_over}
\centering
\begin{tabular}{@{}rrrcrrrcrrr@{}}\toprule
Dim. & \multicolumn{2}{c}{Order} && \multicolumn{3}{c}{Accuracy $\Theta$} && \multicolumn{3}{c}{Computation Time [s]} \\
\cmidrule{2-3} \cmidrule{5-7} \cmidrule{9-11}       
$n$           & $\mathcal{Z}_m$         & $\mathcal{Z}_s$ && HC & HC-o & HC-co && HC & HC-o & HC-co \\ \midrule
$2$ & $2$ & $2$ && $1.000$  & $1.000$ & $1.041$  && $0.004$  & $0.024$ & $0.003$ \\
$2$ & $4$ & $2$ && $1.000$  & $1.000$ & $1.001$  && $0.004$  & $0.044$ & $0.004$ \\
$2$ & $2$ & $4$ && $1.000$  & $1.000$ & $1.025$  && $0.004$  & $0.024$ & $0.004$ \\
$2$ & $4$ & $4$ && $1.000$  & $1.000$ & $1.001$  && $0.004$  & $0.046$ & $0.004$ \\ \hline
$4$ & $2$ & $2$ && $1.000$  & $1.106$ & $1.142$  && $0.005$  & $0.336$ & $0.005$ \\
$4$ & $4$ & $2$ && $1.000$  & $1.060$ & $1.068$  && $0.015$  & $7.649$ & $0.014$ \\
$4$ & $2$ & $4$ && $1.000$  & $1.136$ & $1.221$  && $0.005$  & $0.323$ & $0.004$ \\
$4$ & $4$ & $4$ && $1.000$  & $1.049$ & $1.056$ && $0.015$  & $7.649$ & $0.014$  \\ \hline
$6$ & $2$ & $2$ && $1.000$ & $1.239$ & $1.321$ && $0.016$  & $14.70$  & $0.019$ \\
\bottomrule
\end{tabular}
\end{table*}

As expected, the method HC-o is more precise, but computationally more expensive compared to HC-co. Especially for the six-dimensional problems, the coarse computation was significantly faster. Overall, the difference between the accuracy measure for the under-approximation compared to the over-approximation is rather moderate.


\section{Conclusions} \label{sec:conclusions}

This is the first work that presents a concept for under-approximating as well as over-approximating the Minkowski difference of zonotopes. Although the Minkowski difference can be computed exactly for the halfspace representation of zonotopes, doing so would result in the loss of the generator representation. This would require to perform subsequent computations, such as linear maps or Minkowski addition, using a halfspace representation. For such operations in particular, however, zonotopes in generator representation are very efficient and easily outperform the halfspace representation by several orders of magnitude in high-dimensional spaces. Even though the presented Minkowski difference approaches are computationally more expensive than linear maps and Minkowski addition for zonotopes, they are significantly faster than state-of-the-art algorithms for halfspace representations. It is believed that this new approach can be applied in many of the aforementioned fields, such as computing invariant sets and reachable sets of dynamic systems, robust control, robotic path planning, robust interval regression analysis, and cooperative games for system sizes that would not be possible using polytopes. 

\section*{Acknowledgments}

The author gratefully acknowledges financial support by the European Commission project UnCoVerCPS and justITSELF under the respective grant numbers 643921 and 817629.

\bibliographystyle{plain}
\bibliography{althoff_own,althoff_other}

\end{document}